\documentclass{article}
\usepackage{authblk}
\usepackage{amssymb}
\usepackage{amsmath}
\usepackage{amsfonts}
\usepackage{graphicx}
\usepackage{epstopdf}
\usepackage{epsfig}
\usepackage[pagewise]{lineno}
\usepackage{multicol}
\usepackage{subcaption}
\usepackage{caption}
\usepackage{parskip}
\providecommand{\U}[1]{\protect\rule{.1in}{.1in}}

\newtheorem{theorem}{Theorem}

\newtheorem{algorithm}[theorem]{Algorithm}

\newtheorem{definition}[theorem]{Definition}
\newtheorem{example}[theorem]{Example}

\newtheorem{lemma}[theorem]{Lemma}

\newenvironment{proof}[1][Proof]{\noindent\textbf{#1.} }{\ \rule{0.5em}{0.5em}}

\newlength{\defbaselineskip}
\newcommand{\setlinespacing}[1]%
{\setlength{\baselineskip}{#1 \defbaselineskip}}

\newtheorem{thm}{Theorem}[section]

\newcommand{\Z}{\mathbb Z}

\makeatletter
\newcommand{\mathleft}{\@fleqntrue\@mathmargin0pt}
\newcommand{\mathcenter}{\@fleqnfalse}

\graphicspath{{./images/}}
\usepackage[T1]{fontenc}
\usepackage[utf8]{inputenc}
\usepackage{authblk}
\title
{Extending the ElGamal Cryptosystem to the Third Group of Units of $\Z_{n}$}

\author[1]{Jana Hamza \thanks{jana.hamza115@gmail.com}}
	
\affil[1,5]{Department of Mathematics and Physics, Faculty of Arts and Sciences, Lebanese International University, Bekaa, Lebanon}

\author[2]{Mohammad EL Hindi \thanks{mohammadyhindi98@gmail.com}}
\affil[2]{Department of Mathematics and Computer Science, Faculty of Science, Beirut Arab University, Beirut, Lebanon}

\author[3]{Seifeddine Kadri \thanks{Seifedine.kadry@norof.no}}
\affil[3]{Faculty of Applied Computing and Technology, Norof University College, Kristiansand, Norway}

\author[4]{Therrar Kadri \thanks{TKadri@nwpolytech.ca}}
\affil[4]{Department of Science, Nourthwestern Polytechnic, Grande Prairie, Canada}

\author[5]{Yahya Awad \thanks{yehya.awad@liu.edu.lb}}

 \date{}
\newcommand{\keywords}[1]{\noindent \textbf{Keywords:} #1}

\begin{document}
	
	\maketitle

	\begin{abstract}
	In this paper, we extend the ElGamal cryptosystem to the third group of units of the ring $\Z_{n}$, which we prove to be more secure than the previous extensions. We describe the arithmetic needed in the new setting. We also provide some numerical simulations that shows the security and efficiency of our proposed cryptosystem.
	
	\end{abstract}


\section{Introduction}
ElGamal crptosystem was first introduced by T. ElGamal in \cite{6}. Classically the system was defined on the multiplicative group $\mathbb{Z}_{p}^{\ast }$, the group of integers modulo a prime $p$, which is a cyclic group generated by one of its elements, yet this cryptosystem can work in the setting of any cyclic
group $G$. The intractability of the discrete logarithm problem in the group $G$ is the basis for the security of the generalized ElGamal cryptosystem. Moreover, this group $G$ should be carefully chosen so that the operations of $G$ are relatively easy to apply for efficiency. This cryptosystem has been generalized several times over different groups. For more information about these generalizations, we guide the reader for the following (\cite{8},\cite{15},\cite{16},\cite{13},\cite{14},\cite{12}).

In 2006, El-Kassar and Chehade in \cite{1}  introduced a generalization of the group of units of the ring $\mathbb{Z}_{n}$ denoted by the k$^{th}$ group of units of $\mathbb{Z}_{n}$, $U^{k}(\mathbb{Z}_{n})$. For more information about the k$^{th}$ group of units, see (\cite{1},\cite{11},\cite{8}). The authors in \cite{1} determined all rings $R=\mathbb{Z}_{n}$ having the $2^{nd}$ group of units cyclic. These groups were used as an extension of the ElGamal Cryptosystem given by Haraty et al in \cite{8}.
They examined two cases of $U^{2}\mathbb({Z}_{n})$:

\begin{enumerate}
	\item Both $U(\mathbb{Z}_{n})$ and $U^{2}(\mathbb{Z}_{n})$ are cyclic.
	
	\item $U^{2}(\mathbb{Z}_{n})$ is cyclic, while $U(\mathbb{Z}_{n})$ is not cyclic.
\end{enumerate}

 Kadri and El-Kassar in [9] examined the third group of units of $\mathbb{Z}_{n}$ and determined all rings $\mathbb{Z}
_{n}$ having $U^{3}(\mathbb{Z}_{n})$ cyclic and proposed to extend the ElGamal cryptosystem to these groups in the case when they are cyclic.

In this paper, we extend the ElGamal cryptosystem to the third group of
units of $\mathbb{Z}_{n}$ in the case when they are cyclic. In other words, we modify the ElGamal public key encryption scheme from its classical domain natural to the domain of $U^{3}(\mathbb{Z}_{n})$ by extending the arithmetic needed for the modifications in this
domain.

In Section $2$ we describe the construction of $U^{3}(\mathbb{Z}_{n})$ and some theorems related to this group. Section $3$ provides the description of our proposed cryptosystem. Finally, Section $4$ shows a comparison between our work and some previous results.

\section{Preliminaries}

In this section we give a brief presentation of the Classical ElGamal
public key cryptosystem, and the modified ElGamal cryptosystem in the
setting of the second group of units of the ring of integers modulo $n$.

The basic algorithms for the functioning of this cryptosystem are described in the following three algorithms:

\begin{algorithm}
	Key generation
\end{algorithm}

\begin{enumerate}
	\item Find a generator $\alpha $ of the
	 $%
	\mathbb{Z}
	_{p}^{\ast }$.
	
	\item Select a random integer $a$, $1\leq a\leq p-2$, and compute $a^{\alpha
	}\bmod{p}$.
	
	\item $A$'s public key is $(p,\alpha ,\alpha ^{a})$; $A$'s private key is $a$%
	.
\end{enumerate}

The following algorithm shows how $B$ can encrypt a message $m$ to $A$.

\begin{algorithm}
	Encryption
\end{algorithm}

\begin{enumerate}
	\item Obtain $A$'s authentic public key $(p,\alpha ,\alpha ^{a})$.
	
	\item Represent the message as an integer $m$ in the range $(0,1,...,p-1)$.
	
	\item Select a random integer $k$, $1\leq k\leq p-2$.
	
	\item Compute $\gamma =\alpha ^{k}\bmod{p}$ and $\delta =m\cdot (\alpha
	^{a})^{k}\bmod{p}$.
	
	\item Send the ciphertext $c=(\gamma ,\delta )$ to $A$.
\end{enumerate}

Here is the algorithm that $A$ uses to recover the message $m$.

\begin{algorithm}
	Decryption
\end{algorithm}

\begin{enumerate}
	\item Use the private key $a$ to compute $\gamma ^{p-1-a}\bmod{p}=\gamma
	^{-a}$
	
	\item Recover $m$ by computing $(\gamma ^{-a})\cdot \delta \bmod{p}$.
\end{enumerate}

Now for the modified ElGamal cryptosystem over $U^{2}(\mathbb{Z}_{n})$, as we mentioned before, two cases were considered. \\
\underline{Case 1: $U(\mathbb{Z}_{n})$ and $U^{2}(\mathbb{Z}_{n})$ are cyclic}: The elements of $U^{2}(\mathbb{Z}_{n})$ in this case have the form $U^{2}(\mathbb{Z}_{n})=\{r^{i}\bmod n:\gcd (i,\varphi (n))=1\}$, where $r$ is the
generator of $U(\mathbb{Z}_{n})$.

The extended ElGamal public key cryptosystem over $U^{2}(\mathbb{Z}_{n})$ follows the next four algorithms:

\begin{algorithm}
	Generator of $U^{2}(\mathbb{Z}	_{n})$
\end{algorithm}

\begin{enumerate}
	\item Find a generator $\theta _{1}$ of $U(\mathbb{Z}_{n})$.
	
	\item Write the order of $U^{2}(\mathbb{Z}
	_{n})$ as $p_{1}^{\alpha _{1}}p_{2}^{\alpha _{2}}...p_{k}^{\alpha _{k}}$.
	
	\item Select a random integer $s$, $0\leq s\leq \varphi (n)-1$, $(s,\varphi
	(n))=1$.
	
	\item For $j=1$ to $i$, do:
	
	\begin{enumerate}
		\item[4.1.] Compute $\theta _{1}^{N/p_{j}}\bmod n$.
		
		\item[4.2.] If $\theta _{1}^{s(N/p_{j})}\bmod n\equiv \theta _{1}$, then
		go to step $3$.
	\end{enumerate}
	
	\item Return $s$.
\end{enumerate}

For the key generation, use this algorithm:

\begin{algorithm}
	Key generation
\end{algorithm}

\begin{enumerate}
	\item Find a generator $\theta _{1}$ of $U(\mathbb{Z}_{n})$.
	
	\item Find $s$ using the previous algorithm.
	
	\item Compute the order of $U^{2}(\mathbb{Z}_{n})$ using $\varphi ^{2}(n)$.
	
	\item Select a random integer $a$, $2\leq a\leq \varphi ^{2}(n)-1$, and
	compute $f=s^{a}(\bmod \varphi (n))$.
	
	\item $A$'s public key is $(n,\theta _{1},s,f)$ and $A$'s private key is $a$.
\end{enumerate}

For the encryption of a message $m$, the following algorithm is used:

\begin{algorithm}
	Encryption
\end{algorithm}

\begin{enumerate}
	\item $B$ obtains $A$'s authentic public key $(n,\theta _{1},s,f)$.
	
	\item Represent the message as an integer in $U^{2}(\mathbb{Z}_{n})$.
	
	\item Select a random integer $k$, $2\leq k\leq \varphi ^{2}(n)-1$.
	
	\item Compute $q=s^{k}(\bmod\varphi (n)),r=f^{k}(\bmod\varphi
	(n)), $ $\gamma =\theta ^{k}=\theta _{1}^{q}(\bmod n)$ and $\delta
	=m^{r}(\bmod n)$.
	
	\item Send the cipher text $c=(q,\delta )$ to $A$.
\end{enumerate}

Finally, to decrypt the message, use the next algorithm:

\begin{algorithm}
	Decryption
\end{algorithm}

\begin{enumerate}
	\item Use the private key $a$ to compute $b=\varphi ^{2}(n)-a$.
	
	\item Recover the message by computing $t=q^{b}(\bmod\varphi (n))$ and $%
	\delta ^{t}(\bmod n)$.
\end{enumerate}

\underline{Case $2$: $U^{2}(\mathbb{Z}_{n})$ is cyclic, }\underline{while $U(\mathbb{Z}_{n})$ is not cyclic:} The extended ElGamal public key cryptosystem over $U^{2}(\mathbb{Z}_{n})$ follows the next four algorithms:

\begin{algorithm}
	Generator of $U^{2}(\mathbb{Z}_{n})$
\end{algorithm}

\begin{enumerate}
	\item Find a generator $\theta _{1}$ of $U(\mathbb{Z}_{n})$.
	
	\item Write the order of $U^{2}(\mathbb{Z}
	_{n})$ as $p_{1}^{\alpha _{1}}p_{2}^{\alpha _{2}}...p_{k}^{\alpha _{k}}$.
	
	\item Select a random integer $s$, $0\leq s\leq \varphi (p)-1$, $(s,\varphi
	(p))=1$.
	
	\item For $j=1$ to $i$, do:
	
	\begin{enumerate}
		\item[4.1.] Compute $\theta _{1}^{N/p_{j}}\bmod p$.
		
		\item[4.2.] If $\theta _{1}^{s(N/p_{j})}\bmod p\equiv \theta _{1}$, then
		go to step $3$.
	\end{enumerate}
	
	\item Use the Chinese Remainder Theorem to find $\theta $, and $s$ by
	solving the system of congruencies: $x\equiv 2(\bmod 3)$ and $x\equiv
	\theta _{1}^{s}(\bmod p)$.
	
	\item return $s$.
\end{enumerate}

Now, for the key generation, $A$ uses the following algorithm:

\begin{algorithm}
	Key generation
\end{algorithm}

\begin{enumerate}
	\item Find a generator $\theta _{1}$ of $U(\mathbb{Z}_{p})$.
	
	\item Find $s$ using the previous algorithm.
	
	\item Compute the order of $U^{2}(\mathbb{Z}_{p})$ using $\varphi ^{2}(p)$.
	
	\item Select a random integer $a$, $2\leq a\leq \varphi ^{2}(p)-1$, and
	compute $f=s^{a}(\bmod\varphi (p))$.
	
	\item $A$'s public key is $(p,\theta _{1},s,f)$ and $A$'s private key is $a$.
\end{enumerate}

For $B$ to encrypt a message $m$ for $A$, he uses this algorithm:

\begin{algorithm}
	Encryption
\end{algorithm}

\begin{enumerate}
	\item $B$ obtains $A$'s authentic public key $(n,\theta _{1},s,f)$.
	
	\item Represent the message as an integer in $U^{2}(\mathbb{Z}_{p})$.
	
	\item Select a random integer $k$, $2\leq k\leq \varphi ^{2}(p)-1$.
	
	\item Compute $q=s^{k}(\bmod\varphi (p)),r=f^{k}(\bmod\varphi
	(p)), $ $\gamma =\theta ^{k}=\theta _{1}^{q}(\bmod p)$ and $\delta
	=m^{r}(\bmod p)$.
	
	\item Send the cipher text $c=(q,\delta )$ to $A$.
\end{enumerate}

Finally to decrypt the message $m$, $A$ applies the next algorithm:

\begin{algorithm}
	Decryption
\end{algorithm}

\begin{enumerate}
	\item Use the private key $a$ to compute $b=\varphi ^{2}(p)-a$.
	
	\item Recover the message by computing $t=q^{b}(\bmod \varphi (p))$ and $%
	\delta ^{t}(\bmod p)$.\\
\end{enumerate}

\section{Construction Of $U^{3}(\mathbb{Z}_{n})$}

In this paper, in order to apply the ElGamal cryptosystem on $U^{3}(\mathbb{Z}_{n})$, it must be a cyclic group, so we are concerned 
about the values of $n $ that makes $U^{3}(\mathbb{Z}_{n})$ cyclic.

\begin{lemma}
	$U(\mathbb{Z}_{3^{\alpha }})$, $U(\mathbb{Z}_{\varphi (3^{\alpha })})$, and $U(\mathbb{Z}_{\varphi (\varphi (3^{\alpha }))})$ are cyclic for all $\alpha >0.$
\end{lemma}

\begin{lemma}
	$U(\mathbb{Z}_{2.3^{\alpha }})$, $U(\mathbb{Z}_{\varphi (2.3^{\alpha })})$, and $U(\mathbb{Z}_{\varphi (\varphi (2.3^{\alpha }))})$ are cyclic for all $\alpha >0.$
\end{lemma}

Now, we define the operation that gives the group isomorphic to $U^{3}(\mathbb{Z}_{n})$ as follows:

\begin{thm}\label{isomorphism}
	\label{U3 Isomorphism}Let $U(\mathbb{Z}_{n})$, $U(\mathbb{Z}_{\varphi (n)})$, and $U(\mathbb{Z}_{\varphi (\varphi (n))})$ be cyclic groups. Then  $f:U^{3}(\mathbb{Z}_{n})\longrightarrow \mathbb{Z}_{\varphi (\varphi (\varphi(n)))}$ given by 
	\begin{equation*}
		f(a)=\log_{g_{3}}\log_{g_{2}}
			\log_{g_{1}}a \bmod n
	\end{equation*}
	is an isomorphism, where $g_{1}$, $g_{2}$, and $g_{3}$ are the generators of $U(\mathbb{Z}_{n})$, $U(\mathbb{Z}_{\varphi (n)})$, and $U(\mathbb{Z}_{\varphi (\varphi (n))})$ respectively.
\end{thm}

\begin{proof}
	Let $U(\mathbb{Z}_{n})$, $U(\mathbb{Z}_{\varphi (n)})$, and $U(\mathbb{Z}_{\varphi (\varphi (n))})$ be cyclic groups. Let $g_{1}$ be a generator of $U(\mathbb{Z}_{n})$. Since $U(\mathbb{Z}_{n})$ is cyclic and finite of order $\varphi (n)$, then $U(\mathbb{Z}_{n})\approx \mathbb{Z}_{\varphi (n)}$ by a function $h_{1}:U(\mathbb{Z}_{n})\longrightarrow \mathbb{Z}_{\varphi (n)}$ defined by $h_{1}(a)=log _{g_{1}}a\bmod 
	n$. Now since $U^{2}(\mathbb{Z}_{n})$ is a subset of $U(\mathbb{Z}_{n})$, and $U(\mathbb{Z}_{\varphi(n)})$ is a subset 
	 of $\mathbb{Z}_{\varphi(n)}$, then 
 $h_{1}:U^{2}(\mathbb{Z}_{n})\longrightarrow U(\mathbb{Z}_{\varphi (n)})$ is an isomorphism.
	
	\noindent Let $g_{2}$ be a generator of $U(\mathbb{Z}_{\varphi (n)})$. Since $U(\mathbb{Z}_{\varphi (n)})$ is cyclic and finite of order $\varphi (\varphi (n))$, then $U(\mathbb{Z}_{\varphi (n)})\approx \mathbb{Z}_{\varphi (\varphi (n))}$ by a function $h_{2}:U(\mathbb{Z}_{\varphi (n)})\longrightarrow\mathbb{Z}_{\varphi (\varphi (n))}$ defined by $h_{2}(b)=\log _{g_{2}}b\bmod \varphi (n)$ or $h_{2}\circ h_{1}:U^{2}(\mathbb{Z}_{n})\longrightarrow\mathbb{Z}_{\varphi (\varphi (n))}$ defined by $h_{2}\circ h_{1}(a)=\log _{g_{2}}\log
	_{g_{1}}a\bmod n$.  Now since $U^{3}(\mathbb{Z}_{n})$ is a subgoup of $U^{2}(\mathbb{Z}_{n})$ and $U(\mathbb{Z}_{\varphi{(\varphi(n))}})$ is a subgroup of $\mathbb{Z}_{\varphi{(\varphi(n))}}$, then 
	$h_{2}\circ h_{1}:U^{3}(\mathbb{Z}_{n})\longrightarrow U(\mathbb{Z}_{\varphi (\varphi (n))})$ is an isomorphism.
	
	Now let $g_{3}$ be a generator of $U(\mathbb{Z}_{\varphi (\varphi (n))})$. Since $U(\mathbb{Z}_{\varphi (\varphi (n))})$ is cyclic and finite of order $\varphi(\varphi(\varphi (n)))$, then the function $h_{3}:U(\mathbb{Z}_{\varphi (\varphi(n))})\longrightarrow\mathbb{Z}_{\varphi (\varphi (\varphi (n)))}$ defined by $h_{3}(c)=\log _{g_{3}}c\bmod\varphi (\varphi (n))$ is an isomorphism.
	
	This implies that $h_{3}\circ h_{2}\circ h_{1}(a):U^{3}(\mathbb{Z}_{n})\longrightarrow \mathbb{Z}_{\varphi (\varphi (\varphi (n)))}$ defined by $h_{3}\circ h_{2}\circ
	h_{1}(a)=\log _{g_{3}}\log _{g_{2}}\log _{g_{1}}a\bmod n$ is an isomorphism.
\end{proof}

Now for the construction of $U^{3}(\mathbb{Z}_{n})$, we use the following algorithm:

\begin{enumerate}
	\item Find a generator $g_{1}$ for the group $U(\mathbb{Z}_{n})$.
	
	\item Write each element in $U(\mathbb{Z}_{n})$ as a power of $g_{1}$.
	
	 $U(\mathbb{Z}_{n})=\{g_{1}^{i}\bmod n,0\leq i\leq \varphi (n)\}$.
	
	\item Find a generator $g_{2}$ for the group $U(\mathbb{Z}_{\varphi (n)})$.
	
	\item Write each element in $U(\mathbb{Z}_{\varphi (n)})$ as a power of $g_{2}$. 
	
	$U(\mathbb{Z}_{\varphi (n)})=\{g_{2}^{i}\bmod \varphi (n),0\leq i\leq \varphi
	(\varphi (n))\}$.
	
	\item Find $U^{3}(\mathbb{Z}_{n})=\{g_{1}^{g_{2}^{i}\bmod \varphi (n)}\bmod n,\gcd (i,\varphi
	(\varphi (n)))=1\}$.
\end{enumerate}

The following example shows how to find $U^{3}(\mathbb{Z}_{11})$, and the isomorphic element corresponding to each of its elements in $\mathbb{Z}_{2}$ 
\begin{example}
	A generator $g_{1}$ of $U(\mathbb{Z}_{11})$ is $2$.
	
	We have $U(\mathbb{Z}_{11})=
	\{1=2^{0},2=2^{1},3=2^{8},4=2^{2},5=2^{4},6=2^{9};7=2^{7},8=2^{3},9=2^{6},10=2^{5}\} 
	$ all $\bmod 11$, and $U(\mathbb{Z}_{\varphi (11)})=U(\mathbb{Z}_{10})$. A generator $g_{2}$ of $U(\mathbb{Z}_{10})$ is $3$.
	
	$U(\mathbb{Z}_{10})=\{1=3^{0},3=3^{1},7=3^{3},9=3^{2}\}$ all $\bmod 10$.
	
	\begin{tabular}{ll}
		$U^{3}(\mathbb{Z}_{11})$ & $=\{2^{3^{i}}\bmod 11,\gcd (i,4)=1\}$ \\ 
		& $=\{2^{3^{1}}\bmod 11,2^{3^{3}}\bmod 11\}$ \\ 
		& $=\{2^{3}\bmod 11,2^{7}\bmod 11\}$ \\ 
		& $=\{8,7\}$
	\end{tabular}
	
	$U(\mathbb{Z}_{\varphi (\varphi (11))})=U(\mathbb{Z}_{\varphi (10)})=U(\mathbb{Z}_{4})=\{1,3\}$.
	
	Now to find the isomorphism group of $U^{3}(\mathbb{Z}_{11})$, we find $g_{3}$, the generator of $U(\mathbb{Z}_{4})$, $g_{3}=3$.
	
	\begin{tabular}{ll}
		$f(7)$ & $=\log _{3}\log _{3}\log _{2}7\bmod 11$\\
		& $=\log _{3}\log _{3}\log_{2}2^{7}\bmod 11$ \\ 
		& $=\log _{3}\log _{3}7\bmod 11$\\
		& $=\log _{3}\log _{3}3^{3}\bmod 11 $\\ 
		& $=\log _{3}3\bmod 11$\\
		& $=1$
	\end{tabular}
	
	\begin{tabular}{ll}
		$f(8)$ & $=\log _{3}\log _{3}\log _{2}8\bmod 11$\\
		&$=\log _{3}\log _{3}\log
		_{2}2^{3}\bmod 11$ \\ 
		& $=\log _{3}\log _{3}3\bmod 11$\\
	    & $=\log _{3}1\bmod 11$ \\ 
		& $=0$%
	\end{tabular}
	
	Therefore, $7$ in $U^{3}(\mathbb{Z}_{11})$ is isomorphic to $1$ and $8$ in $U^{3}(\mathbb{Z}_{11})$ is isomorphic to $0$, and thus $U^{3}(\mathbb{Z}_{11})=\{7,8\}\approx \{0,1\}=\mathbb{Z}_{2}$.
\end{example}

 The following Theorem explains how to find a generator of $U^{3}(\mathbb{Z}_{n})$.

\begin{thm}
	Let $g$ be a generator of $U^{3}(\mathbb{Z}_{n})$. Then $g$ has the form $g_{1}^{g_{2}^{g_{3}}(\bmod \varphi (n))}(\bmod n)$, where $g_{1}$, $g_{2}$, and $g_{3}$ are the generators of $U(\mathbb{Z}_{n})$, $U(\mathbb{Z}_{\varphi (n)})$, and $U(\mathbb{Z}_{\varphi (\varphi (n))})$ respectively.
\end{thm}

\begin{proof}
	We have\ from Theorem \ref{U3 Isomorphism} that $U^{3}(\mathbb{Z}_{n})\approx\mathbb{Z}_{\varphi (\varphi (\varphi (n)))}$ by a function $f:U^{3}(\mathbb{Z}_{n})\longrightarrow\mathbb{Z}_{\varphi (\varphi (\varphi (n)))}$ defined by $f(a)=\log _{g_{3}}\log
	_{g_{2}}\log _{g_{1}}a~\bmod n,$ where $f$ is a group isomorphism under addition in $\mathbb{Z}_{\varphi (\varphi (\varphi (n)))}.$ Moreover, if $U^{3}(\mathbb{Z}_{n})$ is cyclic of generator $g,$ then $\mathbb{Z}_{\varphi (\varphi (\varphi (n)))}$ is cyclic of generator $f(g).$
	
	\noindent  However, $(\mathbb{Z}_{\varphi (\varphi (\varphi (n)))},+)$ is cyclic of generator $1,$ then $f(g)=1.$ Therefore, 
	\begin{align*}
	g&=f^{-1}(1)\\
	&=(\log _{g_{3}}\log
	_{g_{2}}\log _{g_{1}})^{-1}(1)\\
	&=\log _{g_{1}}^{-1}\log _{g_{2}}^{-1}(g_{3}^{1})\\
	&=\log
	_{g_{1}}^{-1}(g_{2}^{g_{3}})\\
	&=g_{1}^{g_{2}^{g_{3}}(\bmod \varphi (n))}(%
	\bmod n).
	\end{align*}
	
\end{proof}

\begin{definition}
	Let $f$ be the function defined in Theorem \ref{U3 Isomorphism}. The operations in $(U^{3}(\mathbb{Z}_{n}),\oplus,\otimes)$ are defined as follows:
	\begin{enumerate}
		\item  $x\oplus y=x^{\log_{r}y}\bmod n$, where $r$ is the generator of $U({\mathbb{Z}_{n}})$.
		\item   $x\otimes y= 
		f^{-1}(f(x)+f(y))$.
		\item  $x^{n}=f^{-1}(nf(x))$.
	\end{enumerate}
\end{definition}
\section{ElGamal Cryptosystem over $U^{3}(\mathbb{Z}_{n})$}

The following three algorithms illustrate the ElGamal Cryptosystem over $U^{3}(\mathbb{Z}_{n})$.

For key generation, entity $A$ must do the following:

\begin{algorithm}
	(key generation)
\end{algorithm}

\begin{enumerate}
	\item find a generator $g$ of $U^{3}(\mathbb{Z}_{n}).$
	
	\item select a random integer $b,$ $1\leq b\leq \varphi ^{3}(n).$
	
	\item compute $B=g^{b}$.
	
	\item $A$'s public key is $(g,B)$ and $A$'s private key is $b.$
\end{enumerate}

To encrypt a message $m$ for $A$, entity $B$ must use the following
algorithm:

\begin{algorithm}
	(Encryption)
\end{algorithm}

\begin{enumerate}
	\item obtain $A$'s public key $(g,B).$
	
	\item represent the message as an integer $m$ in $U^{3}(\mathbb{Z}_{n}).$
	
	\item select a random integer $a,$ $1\leq a\leq \varphi ^{3}(n).$
	
	\item compute $s=B^{a}.$
	
	\item compute $A=g^{a}.$ 
	
	\item compute $X=m\otimes s.$ 
	
	\item send the cipher text $c=(A,X).$
\end{enumerate}

To recover the message $m$, entity $A$ uses this algorithm:

\begin{algorithm}
	(decryption)
\end{algorithm}

\begin{enumerate}
	\item use the private key to compute $s=A^{b}.$
	
	\item compute $s^{-1}.$
	
	\item recover the message $m$ by computing $m=X\otimes s^{-1}.$
\end{enumerate}

\begin{thm}
	Given a generator $g$ of $U^{3}(\mathbb{Z}_{n})$. Define $B=g^{b}$, $A=g^{a}$, $s=B^{a}=A^{b}$, and $X=m \otimes s$. If $k\in U^{3}(\mathbb{Z}_{n})$ such that $k=X \otimes s^{-1}$, then $k=m$. 	 
\end{thm}

\begin{proof}
	We have $s=B^{a}=A^{b}=g^{ab}$, and $X=m \otimes s=m\otimes g^{ab}$, then $k=X\otimes s^{-1}=m\otimes g^{ab} \otimes (g^{ab})^{-1}=m$. 
	
\end{proof}

The following example is an application of our cryptosystem.

\begin{example}
	
	Let $n=3^{4}=81$.\\
	By applying Theorem \ref{isomorphism}, we get that $U^{3}(\mathbb{Z}_{81})\approx \mathbb{Z}_{6}$, where 
	\begin{align*}
	5&\approx 4, 23\approx 3, 
	32\approx 0, 
	50\approx 1, 
	59\approx 2, 
	77\approx 5
	\end{align*}
	\underline{Key Generation}
	\begin{enumerate}
		\item $g=50$ is a generator of $U^{3}(\mathbb{Z}_{81})$ 
		
\item select $b=4$
		
	\item	compute 
	        	\begin{align*}
	        	B & =g^{b}
	        	=50^{4}  
	        	 =f^{-1}(4f(50))
	        	=f^{-1}(4)  
	        	 =5%
	            \end{align*}
	 \item Public key is $(50,5)$ and private key is $4$.
	\end{enumerate}
	\underline{Encryption}
	\begin{enumerate}
		\item choose $a=2$
		\item compute  \begin{align*}
			s&=B^{a}
			=5^{2}
			=f^{-1}(2f(5))
			=f^{-1}(2.4\bmod{6})
			=f^{-1}(2)
			=59
		\end{align*}
	    \item compute  \begin{align*}
		A&=g^{a}
		=50^{2}
		=f^{-1}(2f(50))
		=f^{-1}(2)
		=59
	\end{align*} 
\item choose $m=77$.
\item compute \begin{align*}
	X&=m\otimes s
	=77\otimes 59
	=f^{-1}(f(77)+f(59))
	=f^{-1}(5+2 \bmod{6})
	=f^{-1}(1)
	=50.
\end{align*}
\item cipher text is $(59, 50)$.
\end{enumerate}
\underline{Decryption}
\begin{enumerate}
	\item compute 
	\begin{align*}
	s&=A^{b}
	=59^{4}
	=f^{-1}(4f(59))
	=f^{-1}(4.2\bmod{6})
	=f^{-1}(2)
	=59
	\end{align*}
	\item compute 
	\begin{align*}
	s^{-1}&=f^{-1}([f(59)]^{-1})
	=f^{-1}(2^{-1})
	=f^{-1}(4)
	=5
	\end{align*}
	\item compute
	\begin{align*}
	m&=X\otimes s^{-1}
	=50\otimes 5
	= f^{-1}(f(50)+f(5))
	=f^{-1}(1+4 \bmod{6})
	=f^{-1}(5)
	=77.
	\end{align*}
\end{enumerate}
	\end{example}
\section{Efficiency and security of the cyptosystem}
In this section we present a comparative study between the efficiency and security of our cryptosystem and that present in \cite{8}.  
Since the Baby step- Giant step attack algorithm depends basically on the operation $g^{i}$, where $g$ is the generator of the selected group, it was enough for us to compare the compiling time of $g^{i}$ for both groups. We generated our algorithms on  Wolfram Mathematica 12. We used two groups $U^{2}(\mathbb{Z}_{n})$ and $U^{3}(\mathbb{Z}_{m})$ of approximately equal orders (difference between orders is $2$), and after running both programs 50 times, on randomly chosen elements from both groups, the results are presented in Figure 1. The results prove that the timing for each iteration in $U^{3}(\mathbb{Z}_{m})$ was around 30 times that of $U^{2}(\mathbb{Z}_{n})$, which indicates that the iterations are way more complex in our new cryptosystem, and made it way harder to crack the system.
\begin{figure}[hbt!]
	\centering
	\caption{Time comparison between iterations done on algorithms of $U^{2}(\mathbb{Z}_{n})$ and $U^{3}(\mathbb{Z}_{m})$ }
	\includegraphics[width=0.5\textwidth]{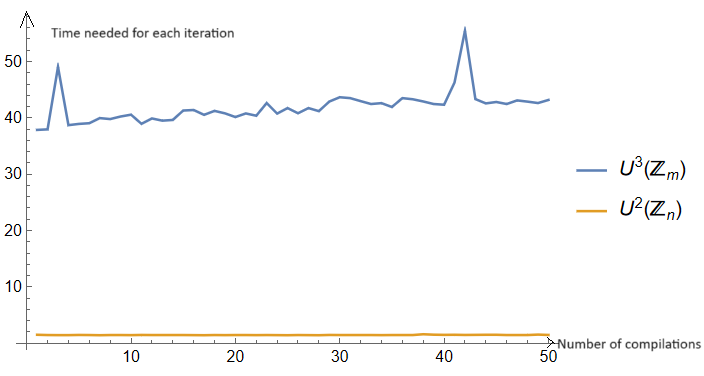}
	\label{xip311}	
\end{figure} 
Note that the blue curve corresponds to timing of $U^{3}(\mathbb{Z}_{m})$, and the orange curve corresponds to that of $U^{2}(\mathbb{Z}_{n})$.

\end{document}